\documentclass[11pt,leqno]{amsart}
\usepackage{amssymb,amsmath,amsthm,amsfonts}
\usepackage{latexsym}
\usepackage{color}
\allowdisplaybreaks

\newcommand{\Hmm}[1]{\leavevmode{\marginpar{\tiny%
$\hbox to 0mm{\hspace*{-0.5mm}$\leftarrow$\hss}%
\vcenter{\vrule depth 0.1mm height 0.1mm width \the\marginparwidth}%
\hbox to 0mm{\hss$\rightarrow$\hspace*{-0.5mm}}$\\\relax\raggedright #1}}}
\newcommand{\nc}{\newcommand}
\nc{\les}{\lesssim}
\nc{\nit}{\noindent}
\nc{\nn}{\nonumber}
\nc{\D}{\partial}
\nc{\diff}[2]{\frac{d #1}{d #2}}
\nc{\diffn}[3]{\frac{d^{#3} #1}{d {#2}^{#3}}}
\nc{\pdiff}[2]{\frac{\partial #1}{\partial #2}}
\nc{\pdiffn}[3]{\frac{\partial^{#3} #1}{\partial{#2}^{#3}}}
\nc{\abs}[1] {\lvert #1 \rvert}
\nc{\cAc}{{\cal A}_c}
\nc{\cE}{{\cal E}}
\nc{\cF}{{\mathcal F}}
\nc{\cP}{{\cal P}}
\nc{\cV}{{\cal V}}
\nc{\cQ}{{\cal Q}}
\nc{\cGin}{{\cal G}_{\rm in}}
\nc{\cGout}{{\cal G}_{\rm out}}
\nc{\cO}{{\cal O}}
\nc{\Lav}{{\cal L}_{\rm av}}
\nc{\cL}{{\cal L}}
\nc{\cB}{{\cal B}}
\nc{\cZ}{{\cal Z}}
\nc{\cR}{{\cal R}}
\nc{\cT}{{\cal T}}
\nc{\cY}{{\cal Y}}
\nc{\cX}{{\cal X}}
\nc{\cXT}{{{\cal X}(T)}}
\nc{\cBT}{{{\cal B}(T)}}
\nc{\vD}{{\vec \mathcal{D}}}
\nc{\efield}{\mathcal{E}}
\nc{\vE}{{\vec \efield}}
\nc{\vB}{{\vec \mathcal{B}}}
\nc{\vH}{{\vec \mathcal{H}}}
\nc{\ty}{{\tilde y}}
\nc{\tu}{{\tilde u}}
\nc{\tV}{{\tilde V}}
\nc{\Pc}{{\bf P_c}}
\nc{\bx}{{\bf x}}
\nc{\bX}{{\bf X}}
\nc{\bXYZ}{{\bf XYZ}}
\nc{\bY}{{\bf Y}}
\nc{\bF}{{\bf F}}
\nc{\bS}{{\bf S}}
\nc{\dV}{{\delta V}}
\nc{\dE}{{\delta E}}
\nc{\TT}{{\Theta}}
\nc{\dPsi}{{\delta\Psi}}
\nc{\order}{{\cal O}}
\nc{\Rout}{R_{\rm out}}
\nc{\eplus}{e_+}
\nc{\eminus}{e_-}
\nc{\epm}{e_\pm}
\nc{\eps}{\varepsilon}
\nc{\vnabla}{{\vec\nabla}}
\nc{\G}{\Gamma}
\nc{\w}{\omega}
\nc{\mh}{h}
\nc{\mg}{g}
\nc{\vphi}{\varphi}
\nc{\tlambda}{\tilde\lambda}
\nc{\be}{\begin{equation}}
\nc{\ee}{\end{equation}}
\nc{\ba}{\begin{eqnarray}}
\nc{\ea}{\end{eqnarray}}

\nc{\g}{\gamma}
\nc{\ol}{\overline}

\newtheorem{theorem}{Theorem}[section]
\newtheorem{lemma}[theorem]{Lemma}
\newtheorem{prop}[theorem]{Proposition}

\newtheorem{defin}[theorem]{Definition}

\nc{\T}{\mathbb T}
\nc{\Z}{\mathbb Z}
\nc{\N}{\mathbb N}
\nc{\pt}{\partial_t}
\nc{\la}{\langle}
\nc{\ra}{\rangle}
\nc{\infint}{\int_{-\infty}^{\infty}}
\nc{\halfwidth}{6.5cm}
\nc{\figwidth}{10cm}

\nc{\nlayers}{L} \nc{\nsectors}{M}
\nc{\indicator}{\mathbf{1}}
\nc{\Rhole}{R_{\rm hole}}
\nc{\Rring}{R_{\rm ring}}
\nc{\neff}{n_{\rm eff}}
\nc{\Frem}{F_{\rm rem}}
\nc{\DD}{\Delta}
\nc{\cD}{\mathcal D}
\nc{\lnorm}{\left\|}
\nc{\rnorm}{\right\|}
\nc{\rnormp}{\right\|_{\ell^{p,\eps}}}
\nc{\rar}{\rightarrow}
\nc{\sgn}{{\rm sign}}
\allowdisplaybreaks
\date{\today}

\begin{document}

\title[Almost Sure Global Well-posedness for Fractional Cubic NLS on $\mathbb{T}$]{Almost Sure Global Well-posedness for Fractional Cubic Schr\"odinger equation on torus}

\author{Seckin Demirbas}

\address{Department of Mathematics \\
University of Illinois \\
Urbana, IL 61801, U.S.A.}

\email{demirba2@illinois.edu}

\begin{abstract}

In \cite{ourpaper}, we proved that $1$-d periodic fractional Schr\"odinger equation with cubic nonlinearity is locally well-posed in $H^s$ for $s>\frac{1-\alpha}{2}$ and globally well-posed for $s>\frac{5\alpha-1}{6}$. In this paper we define an invariant probability measure $\mu$ on $H^s$ for $s<\alpha-\frac{1}{2}$, so that for any $\epsilon>0$ there is a set $\Omega\subset H^s$ such that $\mu(\Omega^c)<\epsilon$ and the equation is globally well-posed for initial data in $\Omega$. We see that this fills the gap between the local well-posedness and the global well-posedness range in almost sure sense for $\frac{1-\alpha}{2}<\alpha-\frac{1}{2}$, i.e. $\alpha>\frac{2}{3}$ in almost sure sense.

\end{abstract}

\maketitle

\section{Introduction}

We consider the cubic periodic fractional Schr\"odinger equation,

\begin{equation}\label{eqn:sch}
\left\{
\begin{array}{l}
iu_{t}+(-\Delta)^{\alpha} u =\gamma|u|^2u, \,\,\,\,  x \in {[0,2\pi]}, \,\,\,\,  t\in \mathbb{R} ,\\
u(x,0)=u_0(x)\in H^{s}([0,2\pi]), \\
\end{array}
\right.
\end{equation}
where $\alpha \in (1/2,1)$ and $\gamma=\pm 1$. The equation is called focusing for $\gamma= 1$ and defocusing for $\gamma=- 1$

On real line, this equation arises as a model in the theory of fractional quantum mechanics, see \cite{laskin}. In \cite{kay}, Kirkpatrick, Lenzmann and  Staffilani derived it as a continuum limit of a model for the interaction of quantum particles on lattice points. Allowing only the nearest points to  interact gives the usual cubic Schr\"odinger equation whereas allowing long range interactions gives rise to the fractional Schr\"odinger equations with paremeter $\alpha$.

For $\alpha=1$, in \cite{bourgain}, Bourgain proved periodic Strichartz's estimates and showed $L^2$ local and global well-posedness for the cubic Schr\"odinger equation. In \cite{bgt}, Burq, Gerard and Tzvetkov noted that this result is sharp since the solution operator is not uniformly continuous on $H^s$ for $s<0$. 

The fractional Schr\"odinger equation on real line was recently studied in \cite{kwon}. For $\alpha\in (1/2,1)$, the equation is less dispersive, so one would not expect to  be able to get local well-posedness on $L^2$ level. Indeed, they proved that there is local well-posedness on $H^s$ for $s\geq \frac{1-\alpha}{2}$. They also showed that the solution operator fails to be uniformly continuous in time for $s<\frac{1-\alpha}{2}$.

 In \cite{ourpaper}, we proved that the periodic fractional equation is locally well-posed in $H^s$, for $s>\frac{1-\alpha}{2}$ using direct $X^{s,b}$ estimates. Further, we proved a Strichartz's estimate of the form,
$$
\|e^{it(-\Delta)^\alpha}f\|_{L^4_{t\in \mathbb{R}}L^4_{x\in \mathbb{T}}}\lesssim \|f\|_{H^s},
$$
 for $s>\frac{1-\alpha}{4}$, which also gives local well-posedness for $s>\frac{1-\alpha}{2}$, using the methods in \cite{cat-wang} and \cite{mypaper}. 
 
  Moreover, we proved in \cite{ourpaper} that the defocusing equation is globally well-posed for $s>\frac{5\alpha+1}{6}$, using Bourgain's high-low frequency decomposition introduced in \cite{highlow}. This method uses the decomposition of the equation into the evolutions of the high and the low frequencies of the initial data. Since the low frequancy part is smooth, its evolution is global due to the conservation of the energy. But the same cannot be said for the high frequancy part. To overcome this problem we showed that the nonlinear evolution of the high frequency part is smoother than the initial data. We should mention that for $\alpha=1$ the smoothing coincides with the smoothing estimate for the NLS that was recently obtained in \cite{talbot}.

  After obtaining these local and global well-posedness results, the natural question that arises is how much we can push the global well-posedness range. For example, the cubic periodic Schr\"odinger equation ($\alpha=1$) in $1$-d is locally well-posed in $L^2$, see  \cite{bourgain}, and with the mass conservation, we know that the equation is globally well-posed. That is, conservation laws on the local well-posedness level may give rise to global well-posedness. But then, one can ask whether we can show that the equation is globally well-posed whenever it is locally well-posed. Although when there is no conservation laws on the local well-posedness level, it is not trivial that this statement is true, we can still make sense of the question in a different way. The idea relies on the intuition that the set of 'bad' initial data, where the solutions of the equation with those initial data, may have arbitrarily large norm, should be negligible. This approach of looking at the problem in an 'almost sure' sense originated from the work of Lebowitz, Rose and Spear, \cite{leb}. They were trying to understand the general behavior of a system containing a large number of particles by looking at the values of the observables by taking averages over certain probability distributions containing only a few parameters instead of looking at the individual initial value problems. With this in mind, they constructed probability measures on Sobolev spaces and proved some basic properties of these measures.

  Later, Bourgain. in \cite{bourgain_measure}, proved that the Schr\"odinger equation with power nonlinearity,
\begin{equation}\label{eqn:sch}
\left\{
\begin{array}{l}
iu_{t}-\Delta u =-|u|^{p-2}u, \,\,\,\,  x \in {[0,2\pi]}, \,\,\,\,  t\in \mathbb{R} ,\\
u(x,0)=u_0(x)\in H^{s}([0,2\pi]), \\
\end{array}
\right.
\end{equation}
where $4<p\leq 6$ is locally well-posed in $H^s$ with $s>0$. But for $0<s<1$ there is no conservation law which would easily allow us to extend the local solutions to global ones. He used the idea of Lebowitz, Rose and Spear to construct a probability measure, also known as Gibbs' measure, on $H^s$ for $s<\frac{1}{2}$ which is invariant under the solution flow. Then he showed that for any $\epsilon>0$, there is global in time $H^s$ norm bounds on the solutions with the initial data in $H^s$ up to a set of measure less than $\epsilon$, i.e. the equation is almost surely globally well-posed in $H^s$ for $0<s<\frac{1}{2}$.

The idea of Gibbs' measures and almost sure global well-posedness later have been used to prove similar results for different equations by \cite{bulut}, \cite{bulut2}, \cite{burq}, \cite{coll}, \cite{oh}, \cite{pav}, \cite{oh2}, \cite{oh3}, \cite{richards} and many others.

Our main result here, is the explicit construction of Gibbs' measure for $1$-d fractional periodic cubic Schr\"odinger equation and the proof of almost sure global well-posedness. More precisely, we define an invariant probability measure $\mu$ on $H^s$, for $s<\alpha-\frac{1}{2}$ such that for any $\epsilon>0$ we can find a set $\Omega\subset H^s$ satisfying $\mu(\Omega^c)<\epsilon$ and the solution to the equation \eqref{eqn:sch} exists globally for all initial data in $\Omega$.
  
  For that, we are going to truncate the equation \eqref{eqn:sch}, and use the idea of invariant measures on finite dimensional Hamiltonian systems. Namely, if we look at the equation,
\begin{equation}\label{trun}
\left\{
\begin{array}{l}
iu^N_{t}+(-\Delta)^{\alpha} u^N =\gamma P_N|u^N|^2u^N,\\
u^N(x,0)=P_Nu_0(x) \\
\end{array}
\right.
\end{equation}
  where $P_N$ is the projection operator onto the first $N$ frequencies, we see that \eqref{trun} is a finite dimensional Hamiltonian system, with the Hamiltonian,
$$
H_N(u)(t)=\frac{1}{2}\sum_{n\leq N}\big||n|^{\alpha} \widehat{u_n(t)}\big|^2-\frac{\gamma}{4}\int_{\mathbb{T}}|\sum_{n\leq N} e^{inx}\widehat{u_n(t)}|^4.
$$

By Liouville's theorem, we know that the Lebesgue measure $\prod_{|n|\leq N} d\widehat{u_n}$ is invariant under the Hamiltonian flow. Thus, by the conservation of the Hamiltonian and the invariance of the Lebesgue measure under the flow, we see that the finite measure,
$$
d\mu_N=e^{-H_N(u)}\prod_{|n|\leq N} d\widehat{u_n},
$$
is invariant under the solution operator, call it S(t).

We see that equation \eqref{eqn:sch} is an infinite dimensional Hamiltonian system on the Fourier side with the Hamiltonian, 
$$
H(u(t))=\frac{1}{2}\sum_{n}\big||n|^{\alpha} \widehat{u_n(t)}\big|^2-\frac{\gamma}{4}\int_{\mathbb{T}}|\sum_{n} e^{inx}\widehat{u_n(t)}|^4=H(u_0),
$$

Then we define the limiting measure $\mu$ on $H^s$ as, 
$$
d\mu=e^{-H(u)}\prod_{n} d\widehat{u_n}=e^{-\frac{1}{2}\sum_{n}\big||n|^{\alpha} \widehat{u_n(t)}\big|^2+\frac{\gamma}{4}\int_{\mathbb{T}}|\sum_{n} e^{inx}\widehat{u_n(t)}|^4}\prod_{n} d\widehat{u_n},
$$
and show that the measure $\mu$ is indeed the weak limit of $\mu_N$.

 To construct this measure $\mu$ on appropriate $H^s$ spaces, we use the theory of Gaussian measures on Hilbert spaces following Zhidkov's arguments in \cite{zhid}, and first define,
$$
dw=e^{-\frac{1}{2}\sum_{n}\big||n|^{\alpha} \widehat{u_n(t)}\big|^2}\prod_{n} d\widehat{u_n}.
$$
Then we show that the measure $\mu$ is absolutely continuous with respect to the Gaussian measure $w$ under certain conditions and finish the proof of almost sure global well-posedness by constructing the set $\Omega\subset H^s$ as stated above. For the second part we will mainly use Bourgain's arguments in \cite{bourgain_measure}.

\section{Notations and Preliminaries}

Recall that for $s\geq 0$, $H^s(\T)$ is defined as a subspace of $L^2$ via the norm
$$
\|f\|_{H^s(\T)}:=\sqrt{\sum_{k\in\Z} \la k\ra^{2s} |\widehat{f}(k)|^2},
$$
where $\la k\ra:=(1+k^2)^{1/2}$ and $\widehat{f}(k)=\frac{1}{2\pi}\int_0^{2\pi}f(x)e^{-ikx} dx$ are the Fourier coefficients of $f$. 
We  use $(\cdot)^+$ to denote $(\cdot)^\epsilon$ for all $\epsilon>0$ with implicit constants depending on $\epsilon$.

   We denote the linear propogator of the equation as $e^{-it(-\Delta)^\alpha}$, which is defined on the Fourier side as $$\widehat{(e^{-it(-\Delta)^\alpha}f)}(n)=e^{it|n|^{2\alpha}}\widehat{f}(n),$$
 and $|\nabla|^{\alpha}$ is defined as $(\widehat{|\nabla|^{\alpha}f)}(n)=  |n|^{ \alpha} \widehat{f}(n)$.
   
    When we say the equation \eqref{eqn:sch} is locally well-posed in $H^s$, we mean that there exist a time $T_{LWP}=T_{LWP}(\|u_0\|_{H^s})$ such that the solution exists and is unique in $X^{s,b}_{T_{LWP}}\subset C([0,T_{LWP}),H^s)$ and depends continuously on the initial data. We say that the the equation is globally well-posed when $T_{LWP}$ can be taken arbitrarily large. Here $X^{s,b}$ denote the Bourgain spaces, and are defined via the restriction in time, of the norm,
 \begin{equation}
 \|u\|_{X^{s,b}}\dot{=}\|e^{it(-\Delta)^\alpha}u\|_{H^b_t(\mathbb{R})H^s_x(\mathbb{T})}=\|\langle \tau-|n|^{2\alpha} \rangle^{b} \langle n \rangle^s\widehat{u}(n,\tau)\|_{L_{\tau}^2 l^2_{(m,n)}},
 \end{equation}
 
 and $\langle x \rangle=(1+|x|^2)^{1/2}$

 By Duhamel Principle, we know that the smooth solutions of \eqref{eqn:sch} satisfy the integral equation
 $$u(t,x)=e^{-it(-\Delta)^\alpha}u_0(x)+i\gamma \int_0^t e^{-i(t-\tau)(-\Delta)^\alpha}|u|^2u(\tau,x)d\tau.$$

We note that along with the Hamiltonian conservation, the equation enjoys mass conservation, namely,
   $$M(u)(t)=\int_{\mathbb{T}}|u(t,x)|^2=M(u)(0).$$

\section{Almost Sure Global Well-posedness}

The main result of this paper is,

\begin{theorem}\label{main}
For $s<\alpha-\frac{1}{2}$ and $\epsilon>0$, there exists an invariant probability measure $\mu$ on $H^s$ such that the equation \eqref{eqn:sch} is globally well-posed for any initial data $u_0 \in \Omega\subset H^s$ such that $\mu(\Omega^c)<\epsilon$ with,
$$
\|u(t)\|_{H^s}\lesssim \Big( \log\big( \frac{1+|t|}{\epsilon} 
\big) \Big)^{s+} .
$$
\end{theorem}

  As we mentioned above, in the proof of this theorem, we first define the finite dimensional measures $\mu_N$, which are invariant under the solution operator of the truncated equation \eqref{trun}, and we define $\mu$ as the weak limit of these measures. But then we have to show how the equation \eqref{eqn:sch} and the truncated equation \eqref{trun} are related, namely,

\begin{lemma}\label{convergence}
Let $A\in \mathbb{R}$ and  $u_0\in H^s$ be such that $\|u_0\|_{H^s}<A$, and assume that the solution, $u_N$, of \eqref{trun} satisfies $\|u_N(t)\|_{H^s}<A \text{ \ for\ \ } t\leq T.$
Then the equation \eqref{eqn:sch} is well-posed in $[0,T]$ and moreover, for any $s'<s$, we have,
\begin{equation}\label{approx}
\|u(t)-u_N(t)\|_{H^{s'}}\leq e^{C_1(1+A)^{C_2}T}N^{s'-s},
\end{equation}
where $C_1$ and $C_2$ independent of $s$.
\end{lemma}

In the proof of this lemma, we use,

\begin{lemma}[see \cite{gin}] For $b,b'$ such that $0\leq b+b'<1$, $0\leq b'<1/2$, then we have $$\Big\|\int_0^t e^{-i(-\Delta)^\alpha(t-\tau)}f(\tau)d\tau\Big\|_{X_T^{s,b}}\lesssim T^{1-b-b'}\|f\|_{X_T^{s,-b'}},$$ for $T\in [0,1].$
\end{lemma} 

\begin{proof}{(Proof of Lemma \eqref{convergence})}
$$
u(t)-u_N(t)=e^{-it(-\Delta)^\alpha}(u_0-P_N u_0)+i\gamma\int\limits_0^te^{-i(t-\tau)(-\Delta)^\alpha}\big( |u|^2u(\tau)-P_N(|u^N|^2u^N)(\tau) \big)d\tau,
$$
and, taking the $ L^\infty( [0,T];H^{s'})$ norms of both sides for $b>\frac{1}{2}$, since $X^{s',b}\subset L^\infty( [0,T], H^{s'})$ for $b>\frac{1}{2}$, we get,
\begin{eqnarray*}
\|u-u_N\|_{ L^\infty( [0,T], H^{s'})}&\leq&\big\|u_0-P_Nu_{0}\big\|_{H^{s'}}+\big\|\int\limits_0^te^{-i(t-\tau)(-\Delta)^\alpha}\big( |u|^2u(\tau)-P_N(|u^N|^2u^N)(\tau) \big)d\tau \big\|_{X^{s',b}}\\
&\leq&\|u_0-P_Nu_{0}\|_{H^{s'}}+(T_{LWP})^{1-b-b'}\big\||u|^2u-P_N|u^N|^2u^N \big\|_{X^{s',-b'}}\\
&\leq&(T_{LWP})^{1-b-b'}\Big(\big\||u|^2u-P_N(|u|^2u)\big \|_{X^{s',-b'}}+\big\|P_N\big(|u|^2u-|u^N|^2u^N\big) \big\|_{X^{s',-b'}}\Big)\\
&&\qquad \qquad+ \|u_0-u_{0,N}\|_{H^{s'}}\\
&\leq&I+II+III,
\end{eqnarray*}
 for $b'<\frac{1}{2}$ such that $b+b'<1$.

The term $III$ is easier to estimate,

\begin{equation*}
III= \big\|\sum_{|n|>N}e^{inx}\widehat{(u_0)}_n\big\|_{H^{s'}}\leq  N^{s'-s}\|u_0\|_{H^s}\leq N^{s'-s}A.
\end{equation*}

For the term $I$, we first observe that $P_N\big( |v|^2v \big)= |v|^2v $ for $v=P_{\frac{N}{3}}u$, from the convolution property of frequency restriction. Then we write,
\begin{eqnarray*}
I&\leq&\big\||u|^2u-P_N(|v|^2v) \big\|_{X^{s',b'}}+\big\|P_N(|v|^2v-|u|^2u)\big \|_{X^{s',-b'}}\\
&=&\big\||u|^2u-|v|^2v \big\|_{X^{s',b'}}+\big\|P_N(|v|^2v-|u|^2u) \big\|_{X^{s',-b'}}\\
&=&I_1+I_2\leq 2 I_1,
\end{eqnarray*}

Estimating term $I_ 1$ using $X^{s,b}$ estimates and local well-posedness theory, see \cite{ourpaper}, we see that,

\begin{eqnarray*}
I_1&\lesssim&  (T_{LWP})^{1-b-b'} \big(\|u\|_{X^{s',b}}+\|v\|_{X^{s',b}} \big)^2\|u-v\|_{X^{s',b}}\\
&\lesssim& (T_{LWP})^{1-b-b'}A^2\|u-P_{\frac{N}{3}}u\|_{X^{s',b}}\\
&\lesssim& (T_{LWP})^{1-b-b'}A^2\|u_0-P_{\frac{N}{3}}u_0\|_{H^{s'}}\\
&\lesssim& (T_{LWP})^{1-b-b'}A^3N^{s'-s}.
\end{eqnarray*}

Thus we get,
\begin{equation*}
I\lesssim  (T_{LWP})^{1-b-b'}A^3N^{s'-s}.
\end{equation*}
Similarly, for the second term we have,
\begin{equation*}
II\lesssim (T_{LWP})^{1-b-b'} \big(\|u\|_{X^{s',b}}+\|u^N\|_{X^{s',b}} \big)^2\|u-u^N\|_{X^{s',b}}\lesssim  (T_{LWP})^{1-b-b'} A^2\|u-u^N\|_{X^{s',b}},
\end{equation*}

and collecting all the terms, we get,
\begin{eqnarray*}
\|u-u_N\|_{X^{s',b}}&\leq& C N^{s'-s}A+C_2 (T_{LWP})^{1-b-b'} A^2\|u-u^N\|_{X^{s',b}}\\
&&\qquad\qquad+ C_1(T_{LWP})^{1-b-b'}A^3N^{s'-s},\\
&\leq& CAN^{s'-s}+\frac{1}{2}\|u-u^N\|_{X^{s',b}}\\
&\leq&2CAN^{s'-s}.
\end{eqnarray*}
for $T_{LWP}$ small enough independent of $N,s$ and $s'$. Repeating this argument, since the implicit constant $C$ can be taken independent of $T_{LWP}$ and $N$, we see that at any $T_{LWP}$ time the norm at most doubles and thus, at time $T$ we get,
$$
\|u-u_N\|_{H^{s'}}\lesssim 2^{\frac{T}{T_{LWP}}}CAN^{s'-s}\approx e^{C'(1+A)^{\delta}T}AN^{s'-s},
$$
which gives the result.
\end{proof}
 Now, we define a probability measure on $H^s$ using the Hamiltonian. For that we will mainly follow Zhidkov's arguments.

   \subsection{Construction of the Measure on $H^s$:}

  First we will fix the notation that we will use for the rest of the paper. Let $F=(-\Delta)^{\alpha-s}$ on $H^s$. We see that $F$ has the orthonormal eigenfunctions $e_n=e^{inx}/\langle n \rangle^{{s}}$ in $H^s$ with the eigenvalues $|n|^{2\alpha-2s}$. We also denote $u_n=(u,e_n)_{H^s}$.

  \begin{defin} 

 A set $M \subset H^s$ is called cylindrical if there exists an integer $k\geq 1$ such that,
$$M=\{ u\in H^s :[u_{-k},\ldots, u_{-2},u_{-1}, u_1, u_2,\ldots , u_k ]\in D \},$$
for a Borel set $D\subset \mathbb{R}^{2k}$.

  \end{defin}

 We denote by $\mathcal{A}$, the algebra containing all such cylindrical sets. Then we define the additive normalized measure $w$ on the algebra $\mathcal{A}$ as follows: For $M\subset \mathcal{A}$, cylindrical,
$$
 w(M)=(2\pi)^{-{k}}\prod_{|n|=1}^k |n|^{\alpha-s}\int_D e^{-\frac{1}{2}\sum_{n=1}^{k}|n|^{2\alpha-2s}|u_n|^2}\prod_{|n|=1}^k du_n.
$$

By the definition of the cylindrical sets, we see that the minimal $\sigma-algebra$ $\overline{\mathcal{A}}$ containing $\mathcal{A}$ is the Borel $\sigma-algebra$, see \cite{zhid}. Although the measure is additive by definition, it doesn't necessarily follow that it is countably additive. Indeed,

\begin{theorem}\label{cont-add}
The Gaussian measure $w$ is countably additive on $\mathcal{A}$ if and only if $\sum_n |n|^{2s-2\alpha}<\infty$, i.e. $s<\alpha-\frac{1}{2}$.
\end{theorem}

\begin{proof}{(cf. \cite{zhid})}
Let $\sum_n |n|^{2s-2\alpha}<\infty$. Wel first show that for any $\epsilon>0$, there exists a compact set $K_\epsilon\subset H^s$ with $w(M)<\epsilon$ for any cylindrical set $M$ such that $M\cap K_\epsilon=\emptyset$.

Let $b_n=|n|^{\tilde \epsilon}$ such that $a=\sum_n |n|^{2s-2\alpha+\tilde \epsilon}<\infty$. Then for an arbitrary $R>0$ take the cylindrical sets of the form,
$$
M=\Big\{ u\in H^s: [u_{-k},\ldots,u_{-2},u_{-1}, u_1,\ldots, u_k]\in D \text{,\quad where }\ \sum_{|n|=1}^k |n|^{\tilde \epsilon}u_n^2>R^2\Big\}.
$$

 Then we see that,
\begin{eqnarray*}
w(M)&=&(2\pi)^{-{k}}\prod_{|n|=1}^k |n|^{\alpha-s}\int_{\sum_{n=1}^k |n|^{\tilde \epsilon}u_n^2>R^2} e^{-\frac{1}{2}\sum_{n=1}^{k}|n|^{2\alpha-2s}|u_n|^2} \prod_{|n|=1}^k du_n\\
&\leq&(2\pi)^{-{k}}\prod_{|n|=1}^k |n|^{\alpha-s}\int_{\mathbb{R}^n}\sum_{n=1}^k \big(\frac{|n|^{\tilde \epsilon}}{R^2}u_n^2\big) e^{-\frac{1}{2}\sum_{n=1}^{k}|n|^{2\alpha-2s}|u_n|^2} \prod_{|n|=1}^k du_n\\
&\leq& R^{-2}\sum_n |n|^{2s-2\alpha+\tilde \epsilon}\\
&=&aR^{-2},
\end{eqnarray*}
here, to pass to the third line we used integration by parts with $f=\frac{-u_n}{|n|^{2\alpha-2s}}$ and $dg=-|n|^{2\alpha-2s}u_ne^{-\frac{1}{2}|n|^{2\alpha-2s}u_n^2}du_n$. Then, for $R>\sqrt{\frac{a}{\epsilon}}$, we have $w(M)<\epsilon$.

  Hence, if we take $K_\epsilon=\{ u\in H^s : \sum_n |n|^{\tilde{\epsilon}} u_n^2\leq R^2\}$, we get the desired compact set.

Now let $A_1\supset A_2\supset \ldots \supset A_m\supset \ldots$ be a sequence of cylindrical sets in $H^s$ such that $\bigcap_{m=1}^\infty A_m=\emptyset$. Then for any $\epsilon>0$ there exists closed cylindrical sets $C_m\subset A_m$ for all $m$ such that $w(A_m/C_m)<\epsilon 2^{-m-2}$. Let $D_m=\bigcap_{k=1}^m C_k$. Then $w(A_m/D_m)\leq w(\bigcup_{k=1}^m (A_k/C_k))<\epsilon/2$. Let $E_m=D_m\cap K_{\epsilon/2}$, then $E_m$'s are compact with $E_m\subset A_m$ and $w(A_m/E_m)<\epsilon$. Since $\bigcap_m A_m=\emptyset$, $\bigcap_m E_m=\emptyset$, and since $(E_m)$ is a nested sequence of compact sets, we see that $E_m=\emptyset$ for all $m>m_0$ for some $m_0\in \mathbb{N}$.

Hence, $w(A_m)<w(E_m)+\epsilon<\epsilon$, for all $m>m_0$. Thus $w(A_m)\rightarrow 0$, i.e. $w$ is countably additive.

For the converse, assume $w$ is countably additive and also $\sum_n |n|^{2s-2\alpha}=\infty$, i.e. $s\geq \alpha-\frac{1}{2}$. Then consider two cases,

\textbf{Case 1: $(s\leq \alpha)$.} In this case we see that $|n|^{2s-2\alpha}\leq 1$ for any $n$. Consider the cylindrical sets of the form,
$$
M_k=\Big\{ u\in H^s : \big| \sum_{|n|=1}^k (u_n^2)-\lambda_k \big|< 2\sqrt {\lambda_k} \Big\},
$$
where $\lambda_k=\sum_{|n|=1}^k |n|^{2s-2\alpha} $.

Then we have,
\begin{eqnarray*}
w(M_k^c)&=&w(\Big\{ u\in H^s : \big| \sum_{|n|=1}^k (u_n^2)-\lambda_k) \big|\geq 2\sqrt {\lambda_k} \Big\})\\
&\leq&\int_{\mathbb{R}^{2n}}\frac{\big(\sum_{|n|=1}^k (u_n^2)-\lambda_k\big)^2}{4\lambda_k} e^{-\frac{1}{2}\sum_{|n|=1}^{k}|n|^{2\alpha-2s}|u_n|^2} \prod_{|n|=1}^k du_n\\
&=&\frac{1}{4\lambda_k}\int_{\mathbb{R}^{2n}}\Big(\big(\sum_{|n|=1}^k u_n^2\big)^2-2\big(\sum_{|n|=1}^k u_n^2\big)\lambda_k+\lambda_k^2\Big) e^{-\frac{1}{2}\sum_{|n|=1}^{k}|n|^{2\alpha-2s}|u_n|^2} \prod_{|n|=1}^k du_n\\
&=&\frac{1}{4\lambda_k}\Big( \big( \lambda_k^2 + 2 \sum_{|n|=1}^k |n|^{4s-4\alpha}\big)-2\lambda_k.\lambda_k+\lambda_k^2 \Big)\\
&\leq&\frac{1}{2}\frac{ \sum_{|n|=1}^k |n|^{4s-4\alpha}}{\lambda_k}\\
&\leq&\frac{1}{2},
\end{eqnarray*}

where, to pass from the third line to the fourth line we used integration by parts again. Since $\lambda_k\rightarrow \infty$ as $k\rightarrow \infty$, there exist balls $B_{\lambda_k-2\sqrt{\lambda_k}}(0)$ of arbitrarily large radii with $w(B_{\lambda_k-2\sqrt{\lambda_k}}(0))\leq w(M_k^c)\leq \frac{1}{2}$, which contradicts with the countably additivity of $w$.

\textbf{Case 2: $(s>\alpha)$.} In this case, for each $n\geq 1$, consider the cylindrical set
$$
M_k=\{ u\in H^s : |u_i|\leq k,\ \  |i|=1,2,...,a_k\},
$$
where $a_k>0$ is an integer. Then by a change of variables, we have,
\begin{eqnarray*}
w(M_k)&=& (2\pi)^{-{a_k}}\prod_{|n|=1}^{a_k} \Big(\int_{-k|n|^{\alpha-s}}^{k|n|^{\alpha-s}}e^{-\frac{1}{2}|u_n|^2} du_n\Big)\\
&\leq&\Big[  (2\pi)^{-{1}}\int_{-k}^{k}e^{-\frac{1}{2}|x|^2} dx \Big]^{a_k},
\end{eqnarray*}
since $s>\alpha$. By choosing $a_k$ large enough, we can take $w(M_k)\leq 2^{-k-1}$ for each $k$ and that $a_k\rightarrow \infty$ as $k\rightarrow \infty$. Then since $\bigcup_{k=1}^\infty M_k=H^s$ and $w(H^s)=1$, since $H^s$ is a cylindrical set with full measure. But then $w(\bigcup_{k=1}^\infty M_k)\leq \sum_{k=1}^\infty w(M_k)\leq \frac{1}{2}$, which is a contradiction. Hence the theorem follows.
\end{proof}

 Now we define the sequence of finite dimensional measures $(w_k)$ as follows: For any fixed $k\geq 1$, we take the $\sigma-algebra$, $\mathcal{A}_n$, of cylindrical sets in $H^s$ of the form $M_k=\{ u\in H^s : [u_{-k},\ldots, u_{-2}, u_{-1},u_1,\ldots, u_k]\in D \}$, for some Borel set $D\subset \mathbb{R}^{2k}$. Then,
$$
w_k(M_k)=(2\pi)^{-{k}}\prod_{|n|=1}^k |n|^{\alpha-s}\int_D e^{-\frac{1}{2}\sum_{|n|=1}^{k}|n|^{2\alpha-2s}|u_n|^2} \prod_{|n|=1}^k du_n.
$$
Hence we get the sequence of finite dimensional, countably additive measures $w_k$ on the $\sigma-algebra$ $\mathcal{A}_k$. We can also extend these measures to the $\sigma-algebra$ $\overline{\mathcal{A}}$ in $H^s$, by setting,
$$
w_k(A)=w_k(A\cap H^s_k),\ \ \text{for } A\in \overline{\mathcal{A}},
$$
where $H^s_k=span(e_{-k}, \ldots, e_{-1}, e_1,\ldots, e_k)$, since $A\cap H^s_k$ is a Borel subset of $H^s_k$ for $A\in \overline{\mathcal{A}}$, see \cite{zhid}.
 
The immediate question is whether or not the infinite dimensional Gaussian measure $w$ and the finite measures $w_k$ are related and the answer is,

\begin{prop}
The sequence $w_k$ converge weakly to the measure $w$ on $H^s$ for $s<\alpha-\frac{1}{2}$  as $k\rightarrow \infty$.
\end{prop}

\begin{proof}{(cf. \cite{zhid})}
First, recall that a sequence of measures $\upsilon_m$ is said to converge to a measure $\upsilon$ weakly on $H^s$ if and only if for any continuous bounded functional $\phi$ on $H^s$, 
$$\int \phi(u)d\upsilon_m(u)\rightarrow \int \phi(u)d\upsilon(u).$$
Also recall that any $\epsilon>0$, if we take $K_\epsilon\subset H^s$ as in the Theorem \eqref{cont-add}, we see that $w(K_\epsilon)>1-\epsilon$ and moreover, $w_m(K_\epsilon)>1-\epsilon$ for all $n\geq 1$. Now let $\phi$ be an arbitrary continuous bounded functional on $H^s$ with $B=\sup_{u\in H^s}\phi(u)$. Then for any $\epsilon>0$ there exists $\delta=\delta(\epsilon)>0$ such that 
\begin{equation}\label{uni-cont}
|\phi(u)-\phi(v)|<\epsilon\ \text{ for any }\ u\in K_\epsilon\text{ and }\ v\in H^s \text{ satisfying }\ \|u-v\|_{H^s}<\delta.
\end{equation}
For any $m$, call $K_m=K_\epsilon\cap H_m^s$. Then by the definition of the measures $w_m$ on $\overline{\mathcal{A}}$, we see that,
\begin{equation}\label{first}
\Big| \int_{H^s} \phi(u)dw_m(u)-\int_{K_m} \phi(u)dw_m(u) \Big|<\epsilon B,
\end{equation}
for any $m\geq1$. Define,
$$
K_{m,\epsilon}=\{ v\in H^s : v=v_1+v_2,\ \  v_1\in H^s_m,\ \  v_2^\perp\in H^s_m,\ \  \|v_2\|_{H^s}<\frac{\delta}{2},\ \   dist(v_1, K_m)<\frac{\delta}{2} \}.
$$
Then, $K_\epsilon\subset K_{m,\epsilon}$ for all sufficiently large $m$'s. Thus, for $m$ large enough,
\begin{equation}\label{second}
\Big| \int_{H^s} \phi(u)dw_m(u)-\int_{K_{m,\epsilon}} \phi(u)dw_m(u) \Big|<\epsilon B.
\end{equation}
We now define the measure $w_m^\perp$ on $(H^s_m)^\perp$ as follows:

For a cylindrical set 
$$
M^\perp= \{ u\in (H^s_m)^\perp: [u_{-m-k},\ldots, u_{-m-2}, u_{-m-1}, u_{m+1},u_{m+2}, \ldots, u_{m+k}]\in F \},
$$
where $F\subset \mathbb{R}^{2k}$ is a Borel set, and,
$$
w_m^\perp(M^\perp)=(2\pi)^{-{k}}\prod_{|n|=m+1}^{m+k} |n|^{\alpha-s}\int_F e^{-\frac{1}{2}\sum_{|n|=m+1}^{m+k}|n|^{2\alpha-2s}|u_n|^2}\prod_{|n|=m+1}^{m+k}du_n .
$$
then we see that $w_m^\perp$ is a probability measure on $(H^s_m)^\perp$ and $w=w_m\otimes w_m^\perp$.

Thus we get,
\begin{equation}
\int_{K_{m,\epsilon}}\phi(u)dw(u)=\int_{u_m\in K_{m,\epsilon}}dw_m(u_m)\int_{u_m^\perp \in K_{m,\epsilon}^\perp(u_m)}\phi(u_m+u_m^\perp)dw_m^\perp(u_m^\perp),
\end{equation}
where, $K_{m,\epsilon}^\perp(u_m)= K_{m,\epsilon}\cap \{ u\in H^s : u=u_m+y, y\in (H^s_m)^\perp \}$. Then by \eqref{uni-cont},
\begin{eqnarray*}
\int_{K_{m,\epsilon}}\phi(u)dw(u)&=&\int_{u_m\in K_{m,\epsilon}}dw_m(u_m)\int_{u_m^\perp \in K_{m,\epsilon}^\perp(u_m)}\big(\phi(u_m+u_m^\perp)-\phi(u_m)\big)+\phi(u_m)dw_m^\perp(u_m^\perp)\\
&\leq&C\epsilon+\int_{u_m\in K_{m,\epsilon}}\phi(u_m)dw_m(u_m),
\end{eqnarray*}
for $C$ independent of $m$ and $\epsilon$.

Hence,
\begin{equation}\label{third}
\int_{K_{m,\epsilon}}\phi(u)dw(u)-\int_{u_m\in K_{m,\epsilon}}\phi(u_m)dw_m(u_m)\leq C\epsilon.
\end{equation}

Therefore, combining \eqref{first}, \eqref{second} and \eqref{third}, we get the result.
\end{proof}

  Now, we show that the measure $\mu$ is absolutely continuous with respect to the Gaussian measure $w$. Recall that,
\begin{eqnarray*}
d\mu_N&=&(2\pi)^{-{N}}\prod_{|n|=1}^N |n|^{\alpha-s}e^{-\frac{1}{2}\sum_{|n|\leq N}\big||n|^{\alpha-s} u_n(t)\big|^2-\frac{\gamma}{4}\int_{\mathbb{T}}|\sum_{|n|\leq N} \frac{e^{inx}}{\langle n \rangle^s}u_n(t)|^4}du_0\prod_{1\leq |n|\leq N} du_n\\
&=&e^{-\frac{\gamma}{4}\int_{\mathbb{T}}|\sum_{|n|\leq N} \frac{e^{inx}}{\langle n \rangle^s}u_n(t)|^4}(2\pi)^{-{N}}\prod_{|n|=1}^N |n|^{\alpha-s}e^{-\frac{1}{2}\sum_{0<|n|\leq N}|n|^{2\alpha-2s} |u_n(t)|^2}du_0\prod_{1\leq |n|\leq N} du_n,
\end{eqnarray*}
and thus, $\mu_N $ is a weighted Gaussian measure.
  
  For the defocusing NLS, since
$$
|u_0|^2\leq \int_{\mathbb{T}}|u|^2 \lesssim \Big( \int_{\mathbb{T}}|u(t)|^4\Big)^{\frac{1}{2}},
$$
 we have,
$$
\int_{u_0\in \mathbb{C}} e^{-\frac{1}{4}\int_{\mathbb{T}}|\sum_{n} \frac{e^{inx}}{\langle n \rangle^s}u_n(t)|^4}du_0\lesssim\int_{u_0\in\mathbb{C}}e^{-\frac{1}{4}|u_0|^4}du_0\lesssim C,
$$
uniformly in $N$. Thus, instead of working with the full measure $\mu_N$ it is enough to work with the measure $w_N$, which is also known as the Wiener measure.
  
  For the focusing NLS, though, we don't have an a priori control over the weight $e^{\frac{1}{4}\int_{\mathbb{T}}|\sum_{n\leq N} e^{inx}\widehat{u_n(t)}|^4}$. We can overcome this problem by using a lemma of Lebovitz et al., see\cite{leb}, which applies an $L^2$ cut-off to the set of initial data, namely,
  
  \begin{lemma}
  $e^{\frac{1}{4}\int |\sum_{ 1\leq |n|\leq N} e^{inx}\widehat{u_n(t)}|^4}\chi_{\{\|u\|_{L^2}\leq B\}}\in L^1(dw_N)$ uniformly in $N$, for all $B<\infty$.
  \end{lemma}
  
  \begin{proof}{(cf. \cite{oh3})}
  \begin{align*}
   \int e^{\frac{1}{4}\int |\sum_{ 1\leq |n|\leq N} e^{inx}\widehat{u_n(t)}|^4}&\chi_{\{\|u\|_{L^2}\leq B\}}dw=\int\limits_{({\int |\sum\limits_{|n|=1 }^N e^{inx}\widehat{u_n(t)}|^4}\leq K)}  e^{\frac{1}{4}\int |\sum_{ 1\leq |n|\leq N} e^{inx}\widehat{u_n(t)}|^4}\chi_{\{\|u\|_{L^2}\leq B\}}dw\\
&+\sum_{i=0}^{\infty} \int\limits_{(\int |\sum\limits_{|n|=1 }^N e^{inx}\widehat{u_n(t)}|^4\in(2^iK,2^{i+1}K])} e^{\frac{1}{4}\int |\sum_{ 1\leq |n|\leq N} e^{inx}\widehat{u_n(t)}|^4}\chi_{\{\|u\|_{L^2}\leq B\}}dw\\
&\leq e^{\frac{1}{4}K^4}+\sum_{i=0}^{\infty}  e^{\frac{1}{4}(2^{i+1}K)^4} w(\{ (\int |\sum\limits_{|n|=1 }^N e^{inx}\widehat{u_n(t)}|^4>2^iK, \|u\|_{L^2}<B \}).  
\end{align*}
 
  Now to estimate the second term on the right hand side, choose $N_0$ dyadic, to be specified later. Now call $N_i=N_0.2^i$ and let $a_i$ be such that $\sum_i a_i=\frac{1}{2}$.

Then,
\begin{equation*}
w(\{ (\|u\|_{L^4}>K, \|u\|_{L^2}<B \})\leq \sum_{i=1}^{\infty}w(\{ (\|P_{\{|n|\approx N_i\}}u\|_{L^4}>a_i K \}),
\end{equation*}
and since by Sobolev embedding we have,
$$
\|P_{\{|n|\approx N_i\}}u\|_{L^4}\lesssim N_i^{\frac{1}{4}}\|P_{\{|n|\approx N_i\}} u\|_{L^2},
$$
  we see that,
\begin{eqnarray*}
w(\{ (\|u\|_{L^4}>K, \|u\|_{L^2}<B \})&\leq& \sum_{i=1}^{\infty}w(\{ \|P_{\{|n|\approx N_i\}}u\|_{L^4}>a_i K \})\\
&\leq&\sum_{i=1}^{\infty}w(\{ \|P_{\{|n|\approx N_i\}}u\|_{L^2}\gtrsim a_i N_i ^{-{\frac{1}{4}}}K \}).
\end{eqnarray*}
Letting $a_i=CN_0^{\epsilon}N_i^{-\epsilon}$ and $N_0$ such that $K \approx N_0^{\frac{1}{4}}B$, i.e.$N_0\approx K^4 B^{-4}$, we get,
\begin{eqnarray*}
w(\{ (\|u\|_{L^4}>K, \|u\|_{L^2}<B \})&\leq& \sum_{i=1}^{\infty}w\Big(\{ \big(\sum_{|n|\approx N_i} |\widehat{u_n}|^2\big)^{\frac{1}{2}}\gtrsim a_i N_i ^{-{\frac{1}{4}}}K \}\Big)\\
&\approx&\sum_{i=1}^{\infty}w\Big(\{ \big(\sum_{|n|\approx N_i} |u_n|^2\big)^{\frac{1}{2}}\gtrsim a_i N_i ^{-{\frac{1}{4}+s}}K \}\Big),
\end{eqnarray*}
and by the estimation of the tail of the Gaussian measure, (cf. \eqref{tail_estimate}), we have,
\begin{eqnarray}
w(\{ (\|u\|_{L^4}>K, \|u\|_{L^2}<B \})&\lesssim&\sum_{i=1}^{\infty} e^{-\frac{1}{4}a_i^2N_i^{(2\alpha-2s)+2s-\frac{1}{2}}K^2}\nonumber\\
&\leq&\sum_{i=1}^{\infty} e^{-\frac{1}{4}N_0^{2\epsilon}N_i^{2\alpha-\frac{1}{2}-2\epsilon}K^2}\nonumber\\
&\leq&\sum_{i=1}^{\infty} e^{-\frac{1}{4}N_0^{2\alpha-\frac{1}{2}}2^{(2\alpha-\frac{1}{2}-2\epsilon)i}K^2}\nonumber\\
&\leq&e^{-\frac{1}{4}}K^2N_0^{2\alpha-\frac{1}{2}}\nonumber\\
&\approx& e^{-\frac{1}{4}K^{2+4(2\alpha-\frac{1}{2})}}B^{4-8s}.
\end{eqnarray}
and collecting terms, we obtain,

 \begin{eqnarray*}
   \int e^{\frac{1}{4}\int |u|^4}\chi_{\{\|u\|_{L^2}\leq B\}}dw&\leq& e^{\frac{1}{4}K^4}+\sum_{i=0}^{\infty}  e^{\frac{1}{4}(2^{i+1}K)^4} w(\{ (\|u\|_{L^4}>2^iK, \|u\|_{L^2}<B \})\\
&\leq& e^{\frac{1}{4}K^4}+\sum_{i=0}^{\infty}  e^{\frac{1}{4}(2^{i+1}K)^4} e^{-\frac{1}{4}(2^iK)^{2+4(2\alpha-\frac{1}{2})}}B^{4-8\alpha}\\
&<&\infty,
\end{eqnarray*}
which proves the lemma. 
  \end{proof}
Moreover, observe that for $\|u\|_{L^2}<B$, we get $|u_0|^2\leq \sum_{n}\frac{|u_n|^2}{\langle n\rangle^{2s}}\leq B^2$. Hence, $L^2$ cut off also restricts $u_0$ to the ball $\{u_0\in \mathbb{C}: |u_0|\leq B\}$, uniformly in $N$. Therefore, combining these two results, we get that the measure $\mu_N$ is a weighted Gaussian measure with weight being uniformly in $L^1$ with respect to the Gaussian measure.

By the construction of the Gaussian measure, we see that for any compact set $E\subset H^s$, we have,
$$
w_N(E\cap H^s_N)\rightarrow w(E).
$$
Thus, using the result above we get,
$$
\mu_N(E\cap H^s_N)\rightarrow  \mu(E).
$$

\begin{proof}{(Proof of Theorem \eqref{main})}
 For the proof of the theorem and the invariance of the measure $\mu$, we follow Bourgain's arguments in \cite{bourgain_measure}. First, for any $\epsilon$ we will construct the sets $\Omega_N\subset H^s$ such that $\mu_N(\Omega_N^c)<\epsilon$ and,
\begin{equation}\label{omega}
\|u^N(t)\|_{H^s}\lesssim \Big( \log\big( \frac{1+|t|}{\epsilon} 
\big) \Big)^{\frac{1}{2}}.
\end{equation}
For that, we fix a large time $T$ and let $[-T_{LWP},T_{LWP}]$ be the local well-posedness interval for the equation \eqref{eqn:sch}. Then consider the set
$$
\Omega^K=\{ u\in H^s_N: \|u\|_{H^s}\leq K \},
$$
where, again, $H^s_N=span\{e_n: |n|\leq N\}$. We see that,
\begin{eqnarray}
w_N((\Omega^K)^c)&=&(2\pi)^{-\frac{N}{2}}\prod_{|n|=1}^N |n|^{\alpha-s}\int\limits_{\{  u\in H^s_N: \|u\|_{H^s}>K  \}} e^{-\frac{1}{2}\sum_{|n|=1}^{N}|n|^{2\alpha-2s}|u_n|^2} \prod_{|n|=1}^N du_n.\nonumber\\
&=&(2\pi)^{-\frac{N}{2}}\prod_{|n|=1}^N |n|^{\alpha-s}\int\limits_{\{  u\in H^s_N: \sum_{|n|\leq N}|u_n|^2>K^2  \}} e^{-\frac{1}{2}\sum_{|n|=1}^{N}|n|^{2\alpha-2s}|u_n|^2} \prod_{|n|=1}^N du_n.\nonumber\\
&=&(2\pi)^{-\frac{N}{2}}\int\limits_{\{ \sum_{|n|\leq N}\frac{|v_n|^2}{\langle n \rangle^{2\alpha-2s}}>K^2  \}} e^{-\frac{1}{2}\sum_{|n|=1}^{N}|v_n|^2}  \prod_{|n|=1}^N dv_n.\nonumber\\
&\leq&(2\pi)^{-\frac{N}{2}}\int\limits_{\{ \sum_{|n|\leq N}|v_n|^2>K^2  \}} e^{-\frac{1}{2}\sum_{|n|=1}^{N}|v_n|^2}  \prod_{|n|=1}^N dv_n\nonumber\\
&=&(2\pi)^{-\frac{N}{2}} \int\limits_{S_{2N}}\int\limits_K^\infty r^{2N-1}e^{-\frac{1}{2}r^2}drdS_{2N}\nonumber\\
&=&(2\pi)^{-\frac{N}{2}} \int\limits_{S_{2N}}\int\limits_K^\infty r\underbrace{r^{2N-2}e^{-\epsilon(r-\epsilon)-\frac{1}{2}\epsilon^2}}_{\leq C}e^{-\frac{1}{2}(r-\epsilon)^2}drdS_{2N}\nonumber\\
&\lesssim&(2\pi)^{-\frac{N}{2}} \int\limits_{S_{2N}}\int\limits_K^\infty re^{-\frac{1}{2}(r-\epsilon)^2}drdS_{2N}\nonumber\\
&\lesssim&(2\pi)^{-\frac{N}{2}} \int\limits_{S_{2N}}\int\limits_K^\infty (r-\epsilon)e^{-\frac{1}{2}(r-\epsilon)^2}drdS_{2N}\nonumber\\
&\leq& e^{-\frac{1}{2}(K-\epsilon)^2}\lesssim e^{-\frac{1}{4}K^2}\label{tail_estimate},
\end{eqnarray}
for $\epsilon$ small enough. Thus, $\mu_N((\Omega^K)^c)\lesssim e^{-\frac{1}{4}K^2}$.

Since $\mu_N$ is invariant under the solution operator, $S_N$ of the truncated equation, if we define the set,
$$
\Omega_N'=\Omega^K\cap S_N^{-1}(\Omega^K)\cap S_N^{-2}(\Omega^K)\cap \ldots \cap S_N^{-\frac{T}{T_{LWP}}}(\Omega^K),
$$
$\Omega_N'$ satisfies the property, $\mu_N((\Omega_N')^c)\leq \frac{T}{T_{LWP}}\mu_N((\Omega^K)^c)<TK^\theta e^{-\frac{1}{4}K^2}$, since the local well-posedness interval $[-T_{LWP},T_{LWP}]$ depends uniformly on the $H^s$ norm of the initial data. Thus if we pick $K= \Big((4+2\theta) \log\big( \frac{T}{\epsilon} 
\big) \Big)^{\frac{1}{2}} $, for $\epsilon$ small we get, 
$$
\mu_N((\Omega_N')^c)<\epsilon,
$$
and by the construction of the set $\Omega_N'$ we have,
$$
\|u^N(t)\|_{H^s}\lesssim \Big( \log\big( \frac{T}{\epsilon} 
\big) \Big)^{\frac{1}{2}},
$$
for all $|t|<T$. Moreover, if we take $T_j=2^j$ and $\epsilon_j=\frac{\epsilon}{2^{j+1}}$, and construct $\Omega_{N,j}$'s, we see that $\Omega_N= \bigcap_{j=1}^\infty \Omega_{N,j}$,  satisfies \eqref{omega}.

Also by the approximation lemma \eqref{approx}, we see that for any $s'<s$ we have,
$$
\|u(t)\|_{H^{s'}}<2A\leq C_{s'} \Big( \log\big( \frac{T}{\epsilon} 
\big) \Big)^{\frac{1}{2}} .
$$
Again by taking an increasing sequence of times, we get,
\begin{equation}\label{norm-bound}
\|u(t)\|_{H^{s'}}\leq C_{s'} \Big( \log\big( \frac{1+|t|}{\epsilon} 
\big) \Big)^{\frac{1}{2}} .
\end{equation}
Hence, if we intersect this result with an increasing sequence of $s<\alpha-\frac{1}{2}$, and taking $\Omega=\bigcap\limits_N \Omega_N$ where $(\Omega_N)$s are defined as above with $\mu_N(\Omega_N^c)<\frac{\epsilon}{2^N}$, we get that $\mu(\Omega)<\epsilon$ and that the solutions to the equation \eqref{eqn:sch} has the norm growth bound,
$$\|u(t)\|_{H^{s}}\leq C_{s} \Big( \log\big( \frac{1+|t|}{\epsilon}\big) \Big)^{\frac{1}{2}},$$
for initial data $u_0\in \Omega$. Moreover, interpolating this bound with,
$$
\|u(t)\|_{L^2}=\|u_0\|_{L^2},
$$
we have,
$$\|u(t)\|_{H^{s}}\leq C \Big( \log\big( \frac{1+|t|}{\epsilon}\big) \Big)^{s+},$$
 which proves theorem \eqref{main}.
\end{proof}

\subsection{Invariance of $\mu$ Under the Solution Flow}

Let $K$ be a compact set and $B_\epsilon$ denote the $\epsilon$ ball in $H^s$. Let $S$ be the flow map for the equation \eqref{eqn:sch} and $S_N$ be the flow map for the equation \eqref{trun}. Then by the weak convergence of the measure,
$$
\mu(S(K)+B_\epsilon)=\lim_{N\rightarrow \infty}\mu_N\big( (S(K)+B_\epsilon)\cap H^s_N \big).
$$
Also by the uniform convergence of the solutions of \eqref{trun} to \eqref{eqn:sch} in $H^{s_1}$ for any $s_1<s$, we get,
$$
S_N(P_N K)\subset S(K)+B_{\epsilon/2},
$$
for $N\geq N_0$ sufficiently large. 
  Then for $\epsilon_1$ small enough,
$$
S_N\big( (K+B_{\epsilon_1})\cap H^s_N \big)\subset S_N(P_N K)+B_{\epsilon/2}\subset S(K)+B_\epsilon.
$$
Hence,
$$
\mu_N\big(S_N\big( (K+B_{\epsilon_1})\cap H^s_N \big)\big)\leq \mu_N\big(S(K)+B_\epsilon\big),
$$
and by the invariance of $\mu_N$, we get,
$$
\mu_N\big( (K+B_{\epsilon_1})\cap H^s_N \big)\leq \mu_N\big(S(K)+B_\epsilon\big),
$$
and letting $N\rightarrow \infty$, by the convergence of the measures $\mu_N$ to $\mu$,
$$
\mu(K)\leq \mu(K+B_{\epsilon_1})\leq \mu(S(K)+B_\epsilon),
$$
which say, by the arbitrariness of $\epsilon$, that $\mu(K)\leq \mu(S(K))$, and by the time reversibility, we also have the inverse inequality and, thus,
$$
\mu(K)=\mu(S(K)),
$$
which gives the invariance of $\mu$ under the solution operator.

\textbf{Acknowledgement:} The author wants to thank Burak Erdo\u{g}an and Nikolaos Tzirakis for their valuable comments and support.

 \end{document}